\documentclass[12pt,a4paper]{article}
\usepackage{amsmath,amssymb,amsthm} 

\newtheorem{prop}{Proposition}

\theoremstyle{remark}
\newtheorem{remark}{Remark}
\newcommand{\beq}{\begin{equation}}
\newcommand{\eeq}{\end{equation}}
\newcommand{\beqnn}{\begin{equation*}}
\newcommand{\eeqnn}{\end{equation*}}
\newcommand{\rd}{\partial}

\newcommand{\CC}{\mathbb{C}}
\newcommand{\PP}{\mathbb{P}}

\newcommand{\ZZ}{\mathbb{Z}}
\newcommand{\frakL}{\mathfrak{L}}
\newcommand{\bsx}{\boldsymbol{x}}
\newcommand{\bst}{\boldsymbol{t}}
\newcommand{\bsalpha}{\boldsymbol{\alpha}}
\newcommand{\bsbeta}{\boldsymbol{\beta}}
\newcommand{\bszero}{\boldsymbol{0}}
\newcommand{\Lbar}{\bar{L}}
\newcommand{\Nbar}{\bar{N}}
\newcommand{\Wbar}{\bar{W}}
\newcommand{\wbar}{\bar{w}}
\newcommand{\tbar}{\bar{t}}
\newcommand{\bstbar}{\bar{\bst}}

\begin{document}

\title{Extended lattice Gelfand-Dickey hierarchy}
\author{Kanehisa Takasaki\thanks{E-mail: takasaki@math.kindai.ac.jp}\\
{\normalsize Department of Mathematics, Kindai University}\\ 
{\normalsize 3-4-1 Kowakae, Higashi-Osaka, Osaka 577-8502, Japan}}
\date{}
\maketitle

\begin{abstract}
The lattice Gelfand-Dickey hierarchy is a lattice version 
of the Gelfand-Dickey hierarchy.  A special case is the lattice 
KdV hierarchy.  Inspired by recent work of Buryak and Rossi, 
we propose an extension of the lattice Gelfand-Dickey hierarchy. 
The extended system has an infinite number of logarithmic flows  
alongside the usual flows.  We present the Lax, Sato and Hirota 
equations and a factorization problem of difference operators 
that captures the whole set of solutions.  The construction of 
this system resembles the extended 1D and bigraded Toda hierarchy, 
but exhibits several novel features as well. 
\end{abstract}

\begin{flushleft}

\end{flushleft}

\newpage

\section{Introduction}

Some years ago, Buryak and Rossi \cite{BR18} pointed out 
that an extension of the lattice KdV hierarchy 
emerges in a variant of cohomological field theory 
of Witten's $r$-spin class.  
The lattice KdV hierarchy (the ``discrete'' KdV hierarchy 
in the terminology of Buryak and Rossi) is a system 
of Lax equations for a difference operator of second order.  
The extended part comprises evolution equations 
with an extra set of time variables and is constructed 
in a rather abstract way.  Finding a more explicit 
construction of these extra flows, as well as 
a higher spin generalization on the basis 
of the Gelfand-Dickey hierarchy, 
is an open problem raised therein.  

In this paper, we present an extension of the lattice 
Gelfand-Dickey (GD for short) hierarchy.  
Although evidence is fragmentary, we believe that this is 
the integrable structure sought for by Buryak and Rossi.  
The lattice GD hierarchy is a reduction 
of the lattice KP hierarchy (aka the discrete KP hierarchy, 
the modified KP hierarchy, etc.) \cite{Dickey-book}.   
Our extended lattice GD hierarchy is obtained 
by adding an infinite number of ``logarithmic flows'' 
to the lattice GD hierarchy.  

Logarithmic flows have been known for the 1D Toda hierarchy 
and its bigraded generalizations.  The name reflects 
the fact that the Lax equations contain the logarithm 
of the Lax operator as a building block.  
These exotic flows were first discovered 
in the large-$N$ limit of a matrix model 
of topological string theory on $\CC\PP^1$ \cite{EY94,EHY95}.  
The Lax equations were accordingly formulated 
in the dispersionless limit of the 1D Toda hierarchy. 
A dispersive version was constructed later \cite{CDZ04} 
and further generalized to the bigraded 
Toda hierarchy \cite{Carlet06}.  The meaning 
of the logarithm of the relevant Lax operator, 
which is a difference operator, was also clarified therein. 
Moreover, Hirota equations for the tau function 
in the presence of logarithmic flows 
were also subsequently obtained 
\cite{Milanov07,Li-etal09,Takasaki10}.  
These extended 1D and bigraded Toda hierarchies 
are used to capture the all-genus Gromov-Witten 
partition functions of $\CC\PP^1$ \cite{DZ03,Milanov06} 
and its two-point orbifolds \cite{MT08,CvdL13}. 

The construction of logarithmic flows 
in the lattice GD hierarchy is mostly parallel 
to our previous work \cite{Takasaki10} 
on the extended 1D Toda hierarchy.  The somewhat unusual 
Lax formalism proposed therein turns out to be applicable 
to the lattice GD hierarchy as well.  Let us stress 
that the construction of logarithmic flows 
by Carlet et al. \cite{CDZ04,Carlet06} is not very suited 
for this purpose.  The bigraded Toda hierarchy is labelled 
by two positive integers $N$ and $\Nbar$. 
The 1D Toda hierarchy is the case of $N = \Nbar = 1$.  
The lattice GD hierarchy amounts to the case of $\Nbar = 0$, 
which is exceptional in the perspectives of 
the bigraded Toda hierarchy.  

This paper is organized as follows.  Section 2 is a review 
of the notion of the lattice KP hierarchy 
and its reduction to the lattice GD hierarchy.  
Section 3 presents the construction of the extended 
lattice GD hierarchy in the language of the Lax and 
Sato equations.  Section 4 deals with the Hirota equations 
for the tau function.  Section 5 is devoted 
to the factorization problem.

\section{Lattice KP and GD hierarchies}

The lattice KP hierarchy may be thought of 
as a subsystem of the 2D Toda hierarchy 
with the second set $\bstbar = (\tbar_k)_{k=0}^\infty$ 
of the time variables $(\bst,\bstbar)$ 
being turned off to $\bstbar = \bszero$ 
(see the recent review \cite{Takasaki18} 
for generalities of the 2D Toda hierarchy). 
We now consider an $\hbar$-dependent version \cite{TT95} 
of these hierarchies.   
Let $s$ be the spatial coordinate therein, 
which is now understood to be a continuous variable, 
and $\Lambda$ be the shift operator 
\[
  \Lambda = e^{\hbar\rd_s}, \quad \rd_s = \rd/\rd s, 
\]
that acts on a function $f(s)$ of $s$ 
as $\Lambda f(s) = f(s + \hbar)$.  

The lattice KP hierarchy comprises the Lax equations 
\beq
  \hbar\frac{\rd L}{\rd t_k} = [B_k,L] = [L,B_k^{-}],\quad 
  k = 1,2,\ldots,
  \label{t-Laxeq}
\eeq
for the difference (or pseudo-difference) Lax operator 
\[
  L = \Lambda + \sum_{n=1}^\infty u_n\Lambda^{1-n}. 
\]
The coefficients $u_n$ depend on $\hbar$ as well as  
$s$ and the time variables $\bst = (t_k)_{k=1}^\infty$.  
The generators $B_k$ and $B_k^{-}$ of the flows are defined 
by the Lax operator as 
\[
  B_k = (L^k)_{\ge 0},\quad 
  B_k^{-} = (L^k)_{<0}, 
\]
where $(\quad)_{\ge 0}$ and $(\quad)_{<0}$ 
denote the projection 
\[
  \left(\sum_{n=-\infty}^\infty a_n\Lambda^n\right)_{\ge 0}
  = \sum_{n\ge 0}a_n\Lambda^n,\quad 
  \left(\sum_{n=-\infty}^\infty a_n\Lambda^n\right)_{<0}
  = \sum_{n<0}a_n\Lambda^n. 
\]
to the non-negative and negative powers of $\Lambda$. 

The Lax equations can be converted to the Sato equations 
\beq
  \hbar\frac{\rd W}{\rd t_k} 
  = (W\Lambda^k W^{-1})_{\ge 0}W - W\Lambda^k 
  = - (W\Lambda^k W^{-1})_{<0}W 
  \label{t-Satoeq}
\eeq
for the dressing operator 
\[
  W = 1 + \sum_{n=1}^\infty w_n\Lambda^{-n}. 
\]
$L$ and $B_k$ are thereby expressed as 
\[
  L = W\Lambda W^{-1},\quad 
  B_k = (W\Lambda^kW^{-1})_{\ge 0}. 
\]
Thus the Sato equations (\ref{t-Satoeq}) imply 
the Lax equations (\ref{t-Laxeq}).  
Conversely, given a Lax operator $L$ satisfying 
(\ref{t-Laxeq}), one can find a solution $W$ 
of (\ref{t-Satoeq}). 

The wave function 
\[
  \Psi = \left(1 + \sum_{n=1}^\infty w_nz^{-n}\right)
         z^{s/\hbar}e^{\xi(\bst,z)/\hbar}, \quad 
  \xi(\bst,z) = \sum_{k=1}^\infty t_kz^k, 
\]
satisfies the auxiliary linear equations 
\beq
  \hbar\frac{\rd\Psi}{\rd t_k} = B_k\Psi,\quad 
  L\Psi = z\Psi. 
  \label{t-lineq}
\eeq
The amplitude part of $\Psi$ can be expressed 
by the tau function $\tau = \tau(\hbar,s,\bst)$ as 
\[
  1 + \sum_{n=1}^\infty w_nz^{-n} 
  = \frac{\tau(\hbar,s,\bst - \hbar[z^{-1}])}{\tau(\hbar,s,\bst)},\quad 
  [z] = \left(\frac{z^k}{k}\right)_{k=1}^\infty. 
\]

Reduction to the lattice GD hierarchy \cite{Dickey-book} 
is induced by the reduction condition 
\beq
  (L^N)_{<0} = 0.
  \label{red-cond}
\eeq
The reduced Lax operator 
\beq
  \frakL = L^N = B_N = \Lambda^N + b_1\Lambda^{N-1} + \cdots + b_N 
  \label{redLaxop}
\eeq
satisfies the Lax equations 
\beq
  \hbar\frac{\rd\frakL}{\rd t_k} = [B_k,\frakL] = [\frakL,B_k^{-}], \quad 
  k = 1,2,\ldots, 
  \label{t-redLaxeq}
\eeq
of the lattice GD hierarchy.  $B_k$ and $B_k^{-}$ can be 
rewritten in terms of fractional powers of $\frakL$ as 
\[
  B_k = (\frakL^{k/N})_{\ge 0},\quad 
  B_k^{-} = (\frakL^{k/N})_{<0}. 
\]
The $N = 2$ case is the lattice KdV hierarchy.  
The existence of two expressions (\ref{t-redLaxeq}) 
of the Lax equations ensures that the reduced form 
(\ref{redLaxop}) of $\frakL$ is preserved by the flows.  
The flows with respect to $t_{kN}$'s become trivial, i.e., 
\beq
  \frac{\rd\frakL}{\rd t_{kN}} = 0, \quad k = 1,2,\ldots. 
\eeq
The Sato equations persist to take the same form 
\beq
  \hbar\frac{\rd W}{\rd t_k} = B_kW - W\Lambda^k = - B_k^{-}W  
  \label{t-redSatoeq}
\eeq
in the reduction.  The associated auxiliary 
linear equations read 
\beq
  \hbar\frac{\rd\Psi}{\rd t_k} = B_k\Psi,\quad 
  \frakL\Psi = z^N\Psi. 
  \label{t-redlineq}
\eeq

\begin{remark}
The usual KP and GD hierarchies are hidden 
in the lattice KP and GD hierarchies.  
The first time variable $t_1$ plays the role 
of spatial coordinate therein.  
Note that the lowest auxiliary linear equation 
\[
  \hbar\frac{\rd\Psi}{\rd t_1} = B_1\Psi, \quad 
  B_1 = \Lambda + b, 
\]
implies the relation 
\beq
  \Lambda\Psi = (\hbar\rd - b)\Psi,\quad 
  \rd = \rd/\rd t_1. 
  \label{Lam-rd-rel}
\eeq
One can thereby rewrite the difference operators 
in the auxiliary linear equations (\ref{t-lineq}) 
and (\ref{t-redlineq}) into differential 
or pseudo-differential operators with respect to $t_1$.  
In particular, the linear difference equation 
$\frakL\Psi = z^k\Psi$ can be thus converted 
to a linear differential equation of the form 
\beq
  Q\Psi = z^N\Psi,\quad 
  Q = (\hbar\rd)^N + a_2(\hbar\rd)^{N-2} + \cdots + a_N, 
\eeq
which may be thought of as the eigenvalue problem 
of the GD hierarchy.  The Lax equations (\ref{t-redLaxeq}) 
turn into Lax equations of the form 
\beq
  \hbar\frac{\rd Q}{\rd t_k} = [A_k,Q], \quad 
  k = 1,2,\ldots. 
\eeq
$A_k$'s are differential operators obtained 
from $B_k$'s just as $Q$ is converted from $\frakL$.  
This is exactly the usual GD hierarchy, in which 
the variable $s$ is left as a parameter.  
Shifting $s$ as $s \to s + \hbar$ amounts to 
B\"acklund-Darboux transformations of the GD hierarchy 
as the relation (\ref{Lam-rd-rel}) suggests. 
\end{remark}

\section{Extended lattice GD hierarchy}

We now extend the lattice GD hierarchy by 
an infinite number of logarithmic flows.  
Let $\bsx = (x_k)_{k=1}^\infty$ be the set 
of time variables of these flows.  

The Lax equations of the extended flows take 
seemingly the same form 
\beq
  \hbar\frac{\rd\frakL}{\rd x_k} 
  = \frakL^k\hbar\frac{\rd\frakL}{\rd s} + [P_k,\frakL] 
  = [\frakL^k\hbar\rd_s + P_k,\frakL], \quad 
  k = 1,2,\ldots, 
  \label{x-Laxeq}
\eeq
as those of the extended 1D/bigraded Toda hierarchy 
\cite{Takasaki10}.  The structure of $P_k$'s, 
however, is different from that case.  
In the present setup, $P_k$'s are difference operators 
of the form 
\[
  P_k = - \left(\frakL^k\hbar\frac{\rd W}{\rd s}W^{-1}\right)_{\ge 0}, 
\]
thus having no negative powers of $\Lambda$.  
These extended Lax equations and the foregoing 
Lax equations (\ref{t-redLaxeq}) constitute 
the extended lattice GD hierarchy. 

These Lax equations can be converted to the Sato equations 
\beq
  \hbar\frac{\rd W}{\rd x_k} 
  = \frakL^k\hbar\frac{\rd W}{\rd s} + P_kW 
  = (\frakL^k\hbar\rd_s + P_k)W - W\Lambda^{kN}\hbar\rd_s 
  \label{x-Satoeq}
\eeq
for the dressing operator $W$.  
The associated wave function is redefined as 
\[
  \Psi = \left(1 + \sum_{n=1}^\infty w_nz^{-n}\right)z^{s/\hbar}
       \exp\left(\hbar^{-1}\xi(\bst,z) + \hbar^{-1}\xi(\bsx,z^N)\log z\right). 
\]
The undressed wave function 
\[
  \Psi_0 = z^{s/\hbar}
    \exp\left(\hbar^{-1}\xi(\bst,z) + \hbar^{-1}\xi(\bsx,z^N)\log z\right)
\]
satisfies the linear equations 
\beq
  \hbar\frac{\rd\Psi_0}{\rd t_k} = \Lambda^k\Psi_0,\quad 
  \hbar\frac{\rd\Psi_0}{\rd x_k} = \Lambda^{kN}\hbar\rd_s\Psi_0.
  \label{vac-lineq}
\eeq
These equations are transformed, via the dressing 
relation $\Psi = W\Psi_0$, to the linear equations 
\beq
  \hbar\frac{\rd\Psi}{\rd t_k} = B_k\Psi,\quad 
  \hbar\frac{\rd\Psi}{\rd x_k} = (\frakL^k\hbar\rd_s + P_k)\Psi 
  \label{tx-lineq}
\eeq
for the dressed wave function$\Psi$.  
These equations and the eigenvalue problem 
$\frakL\Psi = z^N\Psi$ give the auxiliary linear equations 
of the extended lattice GD hierarchy.  

The presence of the logarithmic factors $\log z$ 
and $\log\Lambda = \hbar\rd_s$ in the wave functions 
and the auxiliary linear equations explains the name 
``logarithmic flows''.  Moreover the logarithm 
\[
  \log L = W\log\Lambda W^{-1} = W\hbar\rd_s W^{-1}
\]
of the Lax operator $L$, too, plays a role as follows. 

One can rewrite the product $L^{kN}\log L$ 
of $L^{kN}$ and $\log L$as 
\beq
  L^{kN}\log L 
  = W\Lambda^{kN}\hbar\rd_sW^{-1} 
  = \frakL^k\hbar\rd_s - \frakL^k\hbar\frac{\rd W}{\rd s}W^{-1}, 
\eeq
which splits into two pieces as 
\beq
  L^{kN}\log L = (\frakL^k\hbar\rd_s + P_k) + P_k^{-}, 
\eeq
where 
\[
  P_k^{-} = - \left(\frakL^k\hbar\frac{\rd W}{\rd s}W^{-1}\right)_{<0}. 
\]
The first piece $\frakL^k\hbar\rd_s + P_k$ 
is the generator of the $x_k$-flow 
in the Lax equation (\ref{x-Laxeq}) and 
the Sato equation (\ref{x-Satoeq}).  
Since $L^{kN}\log L$ commute with $\frakL$, 
one can rewrite the Lax equations as 
\beq
  \hbar\frac{\rd\frakL}{\rd x_k} = [\frakL,P_k^{-}]. 
  \label{x-dualLaxeq}
\eeq
In the same sense, the Sato equation turns out 
to have another expression as 
\beq
  \hbar\frac{\rd W}{\rd x_k} = - P_k^{-}W. 
  \label{x-dualSatoeq}
\eeq

The operators $\frakL^k\hbar\rd_s + P_k$ and $P_k^{-}$ may be 
thought of as the $(\quad)_{\ge 0}$ and $(\quad)_{<0}$ parts 
of $\frakL^k\log L = L^{kN}\log L$ 
if these projectors are extended to operators 
of the form $A\rd_s + B$, where $A$ and $B$ are 
genuine difference operators, as 
\[
  (A\rd_s + B)_{\ge 0} = (A)_{\ge 0}\rd_s + (B)_{\ge 0},\quad 
  (A\rd_s + B)_{<0} = (A)_{<0}\rd_s + (B)_{<0}. 
\]
Although being logically problematical, 
such an interpretation seems to be useful 
for the extended 1D/bigraded Toda hierarchies as well. 

\begin{remark}
The Lax operator of the bigraded Toda hierarchy 
of type $(N,\Nbar)$ is a difference operator of the form
\[
  \frakL = \Lambda^N + b_1\Lambda^{N-1} 
           + \cdots + b_{N+\Nbar}\Lambda^{-\Nbar}, 
  \quad b_{N+\Nbar} \not= 0, 
\]
where $N$ and $\Nbar$ are positive integers. 
The 1D Toda hierarchy is the case of $N = \Nbar = 1$.  
On the other hand, the Lax operator (\ref{redLaxop}) 
of the lattice GD hierarchy amounts to the case 
of $\bar{N} = 0$.  Carlet's construction 
of logarithmic flows \cite{Carlet06} is based 
on the averaged (or subtracted) logarithm 
\[
  \mathrm{Log}\frakL 
  = \frac{1}{2N}\log_{+}\frakL + \frac{1}{2\Nbar}\log_{-}\frakL 
  = \frac{1}{2}W\log\Lambda W^{-1} 
    - \frac{1}{2}\Wbar\log\Lambda\Wbar^{-1} 
\]
of $\frakL$, where $\Wbar$ is the second dressing operator 
of the 2D Toda hierarchy, and does not work 
in the $\Nbar = 0$ case.  To overcome this difficulty, 
we have adopted our previous construction of logarithmic flows 
\cite{Takasaki10} with suitable modifications.  
Let us note that our construction for the $\Nbar > 0$ case 
is equivalent to Carlet's construction. 
\end{remark}

\begin{remark}
The building blocks of the Lax, Sato and 
auxiliary linear equations of the $\bsx$-flows 
have the common factor $\hbar$.  Although 
this factor could be removed from these equations, 
we dare to keep it for symmetry with the equations 
of the $\bst$-flows. 
\end{remark}

\begin{remark}
The dual expressions (\ref{x-Laxeq}) and (\ref{x-dualLaxeq}) 
of the Lax equations imply that the reduced form (\ref{redLaxop}) 
of the Lax operator is preserved by the $\bsx$-flows 
as well as the $\bst$-flows.  Moreover one can show, 
in much the same way as in the case 
of the extended bigraded Toda hierarchy \cite{Carlet06}, 
that the Lax equations imply the zero-curvature equations 
\begin{align}
  \hbar\frac{\rd B_k}{\rd t_l} - \hbar\frac{\rd B_l}{\rd t_k} 
  + [B_k,B_l] &= 0, \\
  \hbar\frac{\rd B_k}{\rd x_l} - \hbar\frac{\rd C_l}{\rd t_k} 
  + [B_k,C_l] &= 0, \\
  \hbar\frac{\rd C_k}{\rd x_l} - \hbar\frac{\rd C_l}{\rd x_k} 
  + [C_k,C_l] &= 0, \quad k,l=1,2,\ldots, 
\end{align}
for $B_k$ and $C_k = \frakL^k\hbar\rd_s + P_k$. 
These zero-curvature equations ensure the consistency 
of the Lax equations.  
\end{remark}

\begin{remark}
Equating the coefficients of powers of $\Lambda$ 
in both hand sides of the Lax equations 
(\ref{t-redLaxeq}) and (\ref{x-Laxeq}), 
one can derive a system of evolution equations 
for the coefficients $b_1,\cdots,b_N$ of $\frakL$.  
In particular, the $\Lambda^0$ terms give the evolution equations 
\beq
  \frac{\rd b_N}{\rd t_k} = 0, \quad 
  \frac{\rd b_N}{\rd x_k} = {b_N}^k\frac{\rd b_N}{\rd s}, 
  \quad k = 1,2,\ldots, 
  \label{tx-bNeq}
\eeq
for $b_N$.  
These equations in the case of $N = 2$ agree, 
up to rescaling of the time variables, 
with Buryak and Rossi's result \cite[Theorem4.2]{BR18}. 
This is a main piece of evidence suggesting that 
our extended lattice KdV hierarchy will be 
the integrable structure sought for therein. 
\end{remark}

\section{Hirota equations}

As in the case of the extended 1D/bigraded 
Toda hierarchy \cite{Li-etal09,Takasaki10}, 
Hirota equations for the tau function can be derived 
from bilinear equations for the wave functions. 
Let us first consider those bilinear equations 
for the wave functions of the extended lattice GD hierarchy. 

It is convenient to start from the bilinear equations 
\beq
  \oint\frac{dz}{2\pi i}\Psi(s,\bst,\bsx,z)\Psi^*(s',\bst',\bsx,z) = 0 
  \label{lKP-bilineq}
\eeq
of the lattice KP hierarchy that holds under the condition 
\beq
  (s - s')/\hbar \in \ZZ_{\ge 0} = \{0,1,2,\ldots\} 
  \label{ss'-range}
\eeq
and for two independent copies $\bst = (t_k)_{k=1}^\infty$ 
and $\bst' = (t'_k)_{k=1}^\infty$ of the $\bst$-variables.  
$\oint dz/2\pi i$ denotes the residue 
\[
  \oint\frac{dz}{2\pi i}\sum_{n=-\infty}^\infty f_nz^n = f_{-1} 
\]
of formal Laurent series of $z$.  
$\Psi^* = \Psi^*(s,\bst,\bsx,z)$ is the dual wave function 
\[
  \Psi^* = \left(1 + \sum_{n=1}^\infty w^*_nz^{-n}\right)z^{-s/\hbar}
    \exp\left(-\hbar^{-1}\xi(\bst,z) - \hbar^{-1}\xi(\bsx,z^N)\log z\right)
\]
defined in the same way as the dual wave function 
of the 2D Toda hierarchy.  Namely, the amplitude part 
is determined by the formal adjoint 
\[
  W^* = 1 + \sum_{n=1}^\infty \Lambda^n\cdot w_n 
      = 1 + \sum_{n=1}^\infty w_n|_{s\to s+n\hbar}\Lambda^n
\]
of $W$ as 
\[
  1 + \sum_{n=1}^\infty w^*_n\Lambda^n 
  = \Lambda^{-1}(W^*)^{-1}\Lambda. 
\]
(\ref{lKP-bilineq}) is actually the bilinear equations 
of the 2D Toda hierarchy restricted to the range 
(\ref{ss'-range}) of $s$ and $s'$.  

In the reduction to the lattice GD hierarchy, 
the auxiliary linear equations for $t_{kN}$, $k = 1,2,\ldots$, 
trivialize as 
\[
  \hbar\frac{\rd\Psi}{\rd t_{kN}} = \frakL^k\Psi = z^{kN}\Psi. 
\]
Accordingly, one can insert arbitrary positive powers 
of $z^N$ into the bilinear equation (\ref{lKP-bilineq}) as 
\beq
  \oint\frac{dz}{2\pi i}z^{mN}
  \Psi(s,\bst,\bsx,z)\Psi^*(s',\bst',\bsx,z) = 0. 
  \label{lGD-bilineq}
\eeq
These are bilinear equations that characterize 
the lattice GD hierarchy. 

We can now use the technique developed 
for the extended 1D/bigraded Toda hierarchy \cite{Takasaki10} 
to derive the following bilinear equation for the wave functions 
of the extended lattice GD hierarchy.  

\begin{prop}
The bilinear equation 
\beq
  \oint\frac{dz}{2\pi i}z^{mN}
  \Psi(s - \xi(\bsalpha,z^N),\bst,\bsx+\bsalpha,z) 
  \Psi^*(s' - \xi(\bsbeta,z^N),\bst',\bsx+\bsbeta,z) = 0 
  \label{extlGD-bilineq}
\eeq
holds for $m,(s-s')/\hbar\in\ZZ$ and 
two sets $\bsalpha = (\alpha_k)_{k=1}^\infty$, 
$\bsbeta = (\beta_k)_{k=1}^\infty$ of arbitrary constants. 
\end{prop}

\begin{proof}
One can rewrite the auxiliary linear equations 
with respect to $x_k$ in (\ref{tx-lineq}) as 
\[
  \hbar\left(\frac{\rd}{\rd x_k} - z^{kN}\frac{\rd}{\rd s}\right)\Psi 
  = Q_k\Psi, 
\]
where 
\[
  Q_k = P_k - \hbar\frac{\rd\frakL^k}{\rd s}. 
\]
These equations can be extended to higher orders as 
\[
  \prod_{k\ge 1}\left(\frac{\rd}{\rd x_k} 
                     - z^{kN}\frac{\rd}{\rd s}\right)^{m_k}\Psi 
  = Q_{m_1,m_2,\ldots}\Psi,\quad 
  m_1,m_2,\ldots \in \ZZ_{\ge 0}, 
\]
where $Q_{m_1,m_2,\ldots}$'s are difference operators 
without negative powers of $\Lambda$.  
Now let $Q_{m_1,m_2,\ldots}$ act on both sides 
of the bilinear equation (\ref{lGD-bilineq}).  
Since this difference operator consists 
of only non-negative powers of $\Lambda$, 
the outcome is a linear combination 
of (\ref{lGD-bilineq}) with $m$ being shifted 
to the positive direction.  Consequently, 
the bilinear equation persists to hold under the action 
of $Q_{m_1,m_2,\ldots}$ as 
\[
  \oint\frac{dz}{2\pi i}z^{mN}Q_{m_1,m_2,\ldots}\Psi(s,\bst,\bsx,z)
    \cdot\Psi^*(s',\bst',\bsx,z)= 0, 
\]
equivalently, 
\[
  \oint\frac{dz}{2\pi i}z^{mN}
  \prod_{k\ge 1}\left(\frac{\rd}{\rd x_k} 
                     - z^{kN}\frac{\rd}{\rd s}\right)^{m_k}
  \Psi(s,\bst,\bsx,z)\cdot\Psi^*(s',\bst',\bsx,z) = 0. 
\]
These equations can be packed 
into a generating functional form with 
the auxiliary variables $\alpha_1,\alpha_2,\ldots$ as 
\[
  \oint\frac{dz}{2\pi i}z^{mN} 
  \exp\left(\sum_{k=1}^\infty \alpha_k
    \left(\frac{\rd}{\rd x_k} - z^{kN}\frac{\rd}{\rd s}\right)\right)
  \Psi(s,\bst,\bsx,z)\cdot\Psi^*(s',\bst',\bsx,z) = 0, 
\]
i.e., 
\[
  \oint\frac{dz}{2\pi i}z^{mN} 
  \Psi(s - \xi(\bsalpha,z^N),\bst,\bsx + \bsalpha,z)
  \Psi^*(s',\bst',\bsx,z) = 0. 
\]
Note that the last equation, too, holds 
for $m, (s - s')/\hbar\in \ZZ_{\ge 0}$. 
Therefore one can apply the shift operator 
\[
  \exp\left(- \sum_{k=1}^\infty \beta_kz^{kN}(\rd_s + \rd_{s'})^k\right)
\]
to the integrand to obtain the further deformed 
bilinear equation
\[
  \oint\frac{dz}{2\pi i}z^{mN} 
  \Psi(s - \xi(\bsalpha+\bsbeta,z^N),\bst,\bsx + \bsalpha,z)
  \Psi^*(s'-\xi(\bsbeta,z^N),\bst',\bsx,z) = 0. 
\]
This equation turns into (\ref{extlGD-bilineq}) 
by substituting $\bsalpha \to \bsalpha - \bsbeta$ 
and $\bsx \to \bsx + \bsbeta$. 
\end{proof}

Lastly, we rewrite the bilinear equation 
(\ref{extlGD-bilineq}) in the language of the tau function.  
As in the case of the extended 1D/bigraded Toda hierarchy 
\cite{Li-etal09,Takasaki10}, 
the wave functions are expressed by the tau function 
$\tau = \tau(\hbar,s,\bst,\bsx)$ as 
\[
\begin{aligned}
  \Psi 
  &= \frac{\tau(\hbar,s,\bst - \hbar[z^{-1}],\bsx)}{\tau(\hbar,s,\bst,\bsx)}
     z^{s/\hbar}\exp\left(\hbar^{-1}\xi(\bst,z) 
       + \hbar^{-1}\xi(\bsx,z^N)\log z\right),\\
  \Psi^* 
  &= \frac{\tau(\hbar,s,\bst + \hbar[z^{-1}],\bsx)}{\tau(\hbar,s,\bst,\bsx)}
     z^{-s/\hbar}\exp\left(- \hbar^{-1}\xi(\bst,z) 
        - \hbar^{-1}\xi(\bsx,z^N)\log z\right).
\end{aligned}
\]
Plugging these expressions into (\ref{extlGD-bilineq}), 
we obtain the bilinear equation 
\begin{align}
  &\oint\frac{dz}{2\pi i}z^{mN+(s-s')/\hbar}e^{\xi(\bst-\bst',z)/\hbar}
   \tau(\hbar,s - \xi(\bsalpha,z^N),\bst - \hbar[z^{-1}],\bsx + \bsalpha)
   \notag \\
  &\qquad\mbox{}\times
   \tau(\hbar,s' - \xi(\bsbeta,z^N),\bst' + \hbar[z^{-1}],\bsx + \bsbeta) 
   = 0 
  \label{lGD-Hirotaeq}
\end{align}
that holds for $m, (s - s')/\hbar \in \ZZ_{\ge 0}$. 
Note that the factors 
$\tau(\hbar,s - \xi(\bsalpha,z^N),\bst,\bsx + \bsalpha)$ 
and $\tau(\hbar,s' - \xi(\bsbeta,z^N),\bst',\bsx + \beta)$ 
arising in the denominator has been removed 
by a simple trick \cite[Section 2.4]{Takasaki10}.  
The bilinear equation (\ref{lGD-Hirotaeq}) 
is a generating functional expression 
of an infinite number of Hirota equations.

\section{Factorization problem}

We now turn to an analogue of the factorization problem 
proposed for solving the extended 1D/bigraded 
Toda hierarchy \cite{Takasaki21}.  
Given an invertible difference operator $U$ of infinite order 
that depends on the spatial coordinate and the time variables 
in a particular manner, the problem is to factorize $U$ 
into two difference operators of the form 
\[
  W = 1 + \sum_{n=1}^\infty w_n\Lambda^{-n},\quad 
  \Wbar = \sum_{n=0}^\infty\wbar_n\Lambda^n 
\]
as 
\beq
  U = W^{-1}\Wbar. 
  \label{factorization}
\eeq
In the case of the extended lattice GD hierarchy, 
$U$ is assumed to satisfy the differential equations 
\beq
  \hbar\frac{\rd U}{\rd t_k} = \Lambda^kU,\quad 
  \frac{\rd U}{\rd x_k} = \Lambda^{kN}\frac{\rd U}{\rd s} 
  \label{tx-Ueq}
\eeq
and the algebraic relation 
\beq
  \Lambda^NU = U\varphi(\Lambda), 
  \label{LamU=Uphi}
\eeq
where $\varphi(\Lambda)$ is a difference operator 
of the form 
\[
  \varphi(\Lambda) = \sum_{n=0}^\infty \varphi_n\Lambda^n. 
\]

\begin{prop}
If $W$ and $\Wbar$ solve the factorization problem 
(\ref{factorization}), $W$ is the dressing operator 
of a solution of the extended lattice GD hierarchy.  
Namely, $\frakL= W\Lambda^NW^{-1}$ takes the reduced form 
(\ref{redLaxop}), and $W$ satisfies the Sato equations 
(\ref{t-redSatoeq}) and (\ref{x-Satoeq}). 
Moreover, $\Wbar$ satisfies the Sato equations 
\beq
  \hbar\frac{\rd\Wbar}{\rd t_k} = B_k\Wbar,\quad 
  \hbar\frac{\rd\Wbar}{\rd x_k} 
  = \frakL^k\hbar\frac{\rd\Wbar}{\rd s} + P_k\Wbar 
  \label{WbarSatoeq}
\eeq
of a slightly different form. 
\end{prop}

\begin{proof}
Substituting $U = W^{-1}\Wbar$ in (\ref{LamU=Uphi}) 
yields the algebraic relation 
\beq
  \frakL = W\Lambda^NW^{-1} = \Wbar\varphi(\Lambda)\Wbar^{-1}. 
  \label{frakL-phi-rel}
\eeq
Since the right hand side does not contain 
negative powers of $\Lambda$, $\frakL$ turns out 
to take the reduced form (\ref{redLaxop}). 
In the same way, (\ref{tx-Ueq}) yields the equations 
\[
  \hbar\frac{\rd W}{\rd t_k}W^{-1} + W\Lambda^kW^{-1} 
  = \hbar\frac{\rd\Wbar}{\rd t_k}\Wbar^{-1}
\]
and 
\[
  \hbar\frac{\rd W}{\rd x_k}W^{-1} 
    - W\Lambda^{kN}W^{-1}\hbar\frac{\rd W}{\rd s}W^{-1} 
  = \hbar\frac{\rd\Wbar}{\rd x_k}\Wbar^{-1} 
    - W\Lambda^{kN}W^{-1}\hbar\frac{\rd\Wbar}{\rd s}\Wbar^{-1}, 
\]
which one can rewrite as 
\[
  \hbar\frac{\rd W}{\rd t_k}W^{-1} + \frakL^{k/N} 
  = \hbar\frac{\rd\Wbar}{\rd t_k}\Wbar^{-1}
\]
and 
\[
  \hbar\frac{\rd W}{\rd x_k}W^{-1} 
    - \frakL^k\hbar\frac{\rd W}{\rd s}W^{-1} 
  = \hbar\frac{\rd\Wbar}{\rd x_k}\Wbar^{-1} 
    - \frakL^k\hbar\frac{\rd\Wbar}{\rd s}\Wbar^{-1}. 
\]
Let $B_k$ and $P_k$ be the difference operators 
defined by both hand sides of these equations.  
Since the right hand side have no $(\cdots)_{<0}$ part, 
$B_k$ and $P_k$ are equal to the $(\cdots)_{\ge 0}$ part 
of the left hand side.  One can thus identify 
these operators as 
\[
  B_k 
  = \left(\hbar\frac{\rd W}{\rd t_k}W^{-1} + \frakL^{k/N}\right)_{\ge 0} 
  = (\frakL^{k/N})_{\ge 0} 
\]
and 
\[
  P_k 
  = \left(\hbar\frac{\rd W}{\rd x_k}W^{-1} 
          - \frakL^k\hbar\frac{\rd W}{\rd s}W^{-1}\right)_{\ge 0} 
  = - \left(\frakL^k\hbar\frac{\rd W}{\rd s}W^{-1}\right)_{\ge 0} 
\]
and obtain (\ref{t-redSatoeq}) and (\ref{x-Satoeq}). 
\end{proof}

Although looking arbitrarily given, 
the operator $\varphi(\Lambda)$ 
in the algebraic relation (\ref{LamU=Uphi}) 
has to satisfy a set of differential equations.  
This is a consequence of the differential equations 
(\ref{tx-Ueq}) for $U$. 

\begin{prop}
\beq
  \frac{\rd\varphi(\Lambda)}{\rd t_k} = 0,\quad 
  \frac{\rd\varphi(\Lambda)}{\rd x_k} 
  = \varphi(\Lambda)^k\frac{\rd\varphi(\Lambda)}{\rd s},
  \quad k = 1,2,\ldots. 
  \label{tx-phieq}
\eeq
\end{prop}

\begin{proof}
Differentiating both hand sides of (\ref{LamU=Uphi}) 
by $t_k$ gives 
\[
  \Lambda^N\hbar\frac{\rd U}{\rd t_k} 
  = \hbar\frac{\rd U}{\rd t_k}\varphi(\Lambda) 
    + U\hbar\frac{\rd\varphi(\Lambda)}{\rd t_k}. 
\]
By (\ref{tx-Ueq}), this equation turns into 
\[
  \Lambda^NU\varphi(\Lambda) 
  = \Lambda^NU\varphi(\Lambda) 
    + U\hbar\frac{\rd\varphi(\Lambda)}{\rd t_k}, 
\]
which implies the first equation of (\ref{tx-phieq}). 
One can derive the second equation of (\ref{tx-phieq}) 
by a similar reasoning as follows. 
Differentiating both hand sides of (\ref{LamU=Uphi}) 
by $x_k$ and $s$ yields the two equations 
\[
  \Lambda^N\frac{\rd U}{\rd x_k} 
  = \frac{\rd U}{\rd x_k}\varphi(\Lambda) 
    + U\frac{\rd\varphi(\Lambda)}{\rd x_k},\quad 
  \Lambda^N\frac{\rd U}{\rd s} 
  = \frac{\rd U}{\rd s}\varphi(\Lambda) 
    + U\frac{\rd\varphi(\Lambda)}{\rd s}. 
\]
By (\ref{tx-Ueq}), the first equation can be rewritten as 
\[
  \Lambda^{N+kN}\frac{\rd U}{\rd s} 
  = \Lambda^{kN}\frac{\rd U}{\rd s}\varphi(\Lambda) 
    + U\frac{\rd\varphi(\Lambda)}{\rd x_k}. 
\]
The second equation, multiplied with $\Lambda^{kN}$,  
becomes 
\[
  \Lambda^{N+kN}\frac{\rd U}{\rd s} 
  = \Lambda^{kN}\frac{\rd U}{\rd s}\varphi(\Lambda) 
    + \Lambda^{kN}U\frac{\rd\varphi(\Lambda)}{\rd s}.
\]
Subtracting this equation from the preceding equation 
and using (\ref{LamU=Uphi}), one obtains 
the second equation of (\ref{tx-phieq}). 
\end{proof}

Let us stress that the process 
from the factorization problem (\ref{factorization}) 
to the solution of the extended lattice GD hierarchy 
is reversible.  Namely, given a solution 
of the extended lattice GD hierarchy, 
one can reconstruct the operators $U$ and 
$\varphi(\Lambda)$ as 
\[
  U = W^{-1}\Wbar, \quad 
  \varphi(\Lambda) = \Wbar^{-1}\frakL\Wbar, 
\]
which turn out to satisfy the differential equation 
(\ref{tx-Ueq}) and the algebraic condition (\ref{LamU=Uphi}).  
Thus the factorization problem can capture all solutions 
of the extended lattice GD hierarchy.  

\begin{remark}
The $\Lambda^0$ part of (\ref{tx-phieq}) gives 
the equations 
\beq
  \frac{\rd\varphi_0}{\rd t_k} = 0,\quad 
  \frac{\rd\varphi_0}{\rd x_k} 
  = {\varphi_0}^k\frac{\rd\varphi_0}{\rd s},\quad 
  k = 1,2,\ldots. 
\eeq
This is in accord with the equation (\ref{tx-bNeq}) 
for $b_N$, because $b_N$ is equal to $\varphi_0$ 
as (\ref{frakL-phi-rel}) implies.  
\end{remark}

\begin{remark}
If the coefficients $\varphi_n$ of $\varphi(\Lambda)$ 
are quasi-constants (i.e., $\hbar$-periodic) with 
respect to $s$, $\varphi(\Lambda)$ commutes with $\Lambda$, 
and the differential equations (\ref{tx-Ueq}) can be 
reduced to the differential equations 
\beq
  \frac{\rd\varphi(z)}{\rd t_k} = 0,\quad 
  \frac{\rd\varphi(z)}{\rd x_k} 
  = \varphi(z)^k\frac{\rd\varphi(z)}{\rd s}, 
  \quad k = 1,2,\ldots
\eeq
for the generating function 
\[
  \varphi(z) = \sum_{n=0}^\infty\varphi_nz^n. 
\]
Moreover, the algebraic relation (\ref{frakL-phi-rel}) 
becomes the relation 
\beq
  \frakL = L^N = \varphi(\Lbar) 
  = \sum_{n=0}^\infty\varphi_n\Lbar^n 
  \label{LLbar-rel}
\eeq
connecting $L$ to the second Lax operator 
\[
  \Lbar = \Wbar\Lambda\Wbar^{-1}. 
\]
of the 2D Toda hierarchy hidden behind. 
\end{remark}

\begin{remark}
Algebraic relations of the form (\ref{LLbar-rel}) 
are commonly used as a reduction condition 
in the 2D Toda hierarchy.  The coefficients 
of such a reduction condition, however, are constants.  
For instance, the bigraded Toda hierarchy of type $(N,\Nbar)$ 
is induced by the reduction condition 
\[
  L^N = \Lbar^{-\Nbar}.
\]
In contrast, the coefficients $\varphi_n$ of (\ref{LLbar-rel}), 
as well as those of (\ref{LamU=Uphi}) and (\ref{frakL-phi-rel}), 
are allowed to depend on the time variables 
of the logarithmic flows.  This is a remarkable feature 
of the logarithmic flows in the extended lattice GD hierarchy. 
\end{remark}

\section{Conclusion}

We have constructed an extension of the lattice GD hierarchy 
with an infinite number of logarithmic flows. 
This system is not a special case 
of the extended bigraded Toda hierarchy \cite{Carlet06} 
(even if the second set of the time variable $(\bst,\bstbar,\bsx)$ 
therein is turned off to $\bstbar = \bszero$).  
It is meaningless to let $\Nbar = 0$ naively 
in the construction of the extended bigraded Toda hierarchy 
of type $(N,\Nbar)$.  Therefore we return 
to the slightly different Lax formalism proposed 
in our previous work \cite{Takasaki10}, and modify 
the generators $P_k$ of logarithmic flows 
to fit into the $\Nbar = 0$ case.  This leads 
to a consistent system of Lax, Sato and Hirota equations 
as well as a factorization problem that captures 
the whole set of solutions of this system. 

The extended lattice GD hierarchy exhibits novel features.  
Firstly, the last term $b_N$ of the Lax operator $\frakL$ 
satisfies the simple evolution equations (\ref{tx-bNeq}) 
that are closed within themselves.  These equations 
will correspond to the extended sector 
of cohomological field theory studied 
by Buryak and Rossi \cite{BR18}. 
Secondly, the operator $U$ in the factorization problem 
(\ref{factorization}) is required to satisfy 
the algebraic condition (\ref{LamU=Uphi}) of an unusual form.  
In particular, the coefficients $\varphi_n$ 
of the operator $\varphi(\Lambda)$ on the right hand side 
are allowed to depend on the time variables 
of the logarithmic flows.  Last but not least, 
the structure of the Lax and Hirota equations 
is different from the extended bigraded Toda hierarchy.

\subsection*{Acknowledgements}

This work is partly supported by the JSPS Kakenhi Grant 
JP18K03350 and JP21K03261.


\begin{thebibliography}{99}

\bibitem{BR18}
A.~Buryak and P.~Rossi, 
Extended r-spin theory in all genera and the discrete KdV hierarchy, 
Adv. Math. {\bf 386} (2021), 107794. 

\bibitem{Dickey-book}
L.~A.~Dickey, 
Soliton equations and Hamilton systems, 2nd edition, 
World Scientific, Singapore, 2003.

\bibitem{EY94}
T.~Eguchi and S.-K.~Yang, 
The topological $CP^1$ model and the large-$N$ matrix integral, 
Mod. Phys. Lett. {\bf A9} (1994), 2893--2902.

\bibitem{EHY95}
T.~Eguchi, K.~Hori, and S.-K.~Yang, 
Topological $\sigma$ models and large-$N$ matrix integral, 
Internat. J. Mod. Phys. {\bf A10} (1995), 4203--4224.

\bibitem{CDZ04}
G.~Carlet, B.~Dubrovin and Y.~Zhang, 
The extended Toda hierarchy,
Moscow Math. J. {\bf 4} (2004), 313-332 and 534. 

\bibitem{Carlet06}
G.~Carlet, 
The extended bigraded Toda hierarchy, 
J. Phys. A: Math. Gen. {\bf 39} (2006), 9411--9435. 

\bibitem{Milanov07}
T.~E.~Milanov, Hirota quadratic equations 
for the extended Toda hierarchy, 
Duke Math. J. {\bf 138} (2007), 161--178.  

\bibitem{Li-etal09}
C.-Z.~Li, J.-S.~He, K.~Wu and Y.~Cheng, 
Tau function and Hirota bilinear equations 
for the extended bigraded Toda hierarchy, 
J. Math. Phys. {\bf 51} (2010), 043514. 

\bibitem{Takasaki10}
K.~Takasaki, 
Two extensions of 1D Toda hierarchy, 
J. Phys. A: Math. Theor. {\bf 43} (2010), 434032. 

\bibitem{DZ03}
B.~Dubrovin and Y.~Zhang, 
Virasoro symmetries of the extended Toda hierarchy, 
Comm. Math. Phys. {\bf 250} (2004), 161--193. 

\bibitem{Milanov06}
T.~E.~Milanov, 
Gromov-Witten theory of $\CC P^1$ and integrable hierarchies, 
arXiv:math-ph/0605001. 

\bibitem{MT08}
T.~E.~Milanov and H.-H.~Tseng, 
The spaces of Laurent polynomials, Gromov-Witten theory 
of $\PP^1$-orbifolds, and integrable hierarchies, 
J. Reine Angew. Math. {\bf 2008} (2008), no. 622, 189--235. 

\bibitem{CvdL13}
G.~Carlet and J.~van~de~Leur, 
Hirota equations for the extended bigraded Toda hierarchy 
and the total descendent potential of $\CC P^1$ orbifolds, 
J. Phys. A: Math. Theor. {\bf 46} (2013), 405205. 

\bibitem{Takasaki18}
K.~Takasaki, 
Toda hierarchies and their applications, 
J. Phys. A: Math. Theor. {\bf 51} (2018) 203001. 

\bibitem{TT95}
K.~Takasaki and T.~Takebe, 
Integrable hierarchies and dispersionless limit, 
Rev. Math. Phys. {\bf 7} (1995), 743--808. 

\bibitem{Takasaki21}
K.~Takasaki, 
Dressing operators in equivariant Gromov-Witten theory 
of $\CC\PP^1$, 
J. Phys. A: Math. Theor. {\bf 54} (2021), 35LT02. 

\end{thebibliography}
\end{document}